\documentclass[10pt,twocolumn]{article}

\usepackage[top=0.75in, bottom=1in, left=0.75in, right=0.75in]{geometry}

\usepackage{cite}
\usepackage{amsmath,amssymb,amsfonts,amsthm}
\usepackage{algorithm}
\usepackage{algorithmic}
\usepackage{graphicx}
\usepackage{textcomp}
\usepackage{xcolor}
\usepackage[hidelinks]{hyperref}
\usepackage{booktabs}
\usepackage{array}
\usepackage{multirow}
\usepackage{listings}
\usepackage{authblk}
\usepackage{titlesec}
\usepackage{float}

\titleformat{\section}
  {\normalfont\large\bfseries}{\thesection}{1em}{}
\titleformat{\subsection}
  {\normalfont\normalsize\bfseries}{\thesubsection}{1em}{}
  
\def\BibTeX{{\rm B\kern-.05em{\sc i\kern-.025em b}\kern-.08em
    T\kern-.1667em\lower.7ex\hbox{E}\kern-.125emX}}

\newtheorem{theorem}{Theorem}

\newtheorem{property}{Property}

\lstdefinelanguage{Solidity}{
  morekeywords={pragma,solidity,contract,struct,mapping,address,uint,uint8,uint64,uint256,string,bytes,bytes32,bool,view,returns,public,private,modifier,event,emit,constructor,require,if,revert,memory,storage,external,immutable,internal,abs},
  sensitive=true,
  morecomment=[l]{//},
  morestring=[b]"
}
\title{\Large\bfseries A Patient-Centric Blockchain Framework for Secure Electronic Health Record Management: Decoupling Data Storage from Access Control}

\author[1]{Tanzim Hossain Romel}
\author[1]{Kawshik Kumar Paul}
\author[1]{Tanberul Islam Ruhan}
\author[1]{Maisha Rahman Mim}
\author[1]{Abu Sayed Md. Latiful Hoque}

\affil[1]{Department of Computer Science \& Engineering, Bangladesh University of Engineering \& Technology, Dhaka, Bangladesh}

\date{}

\begin{document}
\maketitle

\begin{abstract}
We present a patient-centric architecture for electronic health record (EHR) sharing that separates content storage from authorization and audit. Encrypted FHIR resources are stored off-chain; a public blockchain records only cryptographic commitments and patient-signed, time-bounded permissions using EIP-712. Keys are distributed via public-key wrapping, enabling storage providers to remain honest-but-curious without risking confidentiality. We formalize security goals (confidentiality, integrity, cryptographically attributable authorization, and auditability of authorization events) and provide a Solidity reference implementation deployed as single-patient contracts. On-chain costs for permission grants average \textbf{78,000 gas} (L1), and end-to-end access latency for 1 MB records is \textbf{0.7--1.4 s} (mean values for S3 and IPFS respectively), dominated by storage retrieval. Layer-2 deployment reduces gas usage by 10--13$\times$, though data availability charges dominate actual costs. We discuss metadata privacy, key registry requirements, and regulatory considerations (HIPAA/GDPR), demonstrating a practical route to restoring patient control while preserving security properties required for sensitive clinical data.
\end{abstract}


\section{Introduction}

The digitization of healthcare has transformed medical practice, enabling evidence-based decision-making and population health management at unprecedented scales. However, health records remain trapped in organizational silos with incompatible systems. When patients seek care from multiple providers or relocate, critical medical history becomes inaccessible, leading to duplicated tests, adverse drug interactions, and suboptimal treatment decisions.

\subsection{The Centralization Problem}

Contemporary health information exchange architectures exhibit three fundamental weaknesses. First, they create single points of failure where system compromise can affect millions of patient records. Major healthcare data breaches have exposed the medical information of hundreds of millions of individuals~\cite{vest2010}. Second, centralized systems require patients to trust intermediary organizations with unfettered access. While policies constrain behavior, insider threats persist, and audit logs maintained by audited entities offer limited assurance. Third, patients exercise minimal control over sharing, contradicting principles of autonomy and informed consent.

\subsection{Blockchain as an Architectural Primitive}

Blockchain technology addresses specific weaknesses through replicated, append-only ledgers where transactions are cryptographically verified rather than institutionally authorized. However, naive blockchain application introduces problems: storing protected health information on public blockchains violates privacy through transparency and immutability, and transaction costs make large document storage impractical.

The key insight is architectural separation: encrypted records reside in off-chain storage optimized for large objects, while blockchain serves exclusively as authorization layer and integrity mechanism. This exploits complementary strengths while avoiding respective weaknesses.

\subsection{Deployment Model}

Our architecture deploys \textbf{one contract per patient} rather than a multi-tenant registry. This design choice provides strong isolation between patients' data, simplifies permission management, and aligns with patient sovereignty principles. While this increases deployment costs (one-time contract creation), it eliminates cross-patient vulnerabilities and simplifies auditing. Healthcare institutions can deploy patient contracts on their behalf with appropriate delegation mechanisms. The contract does not use the ERC-721 token standard, instead implementing a simpler patient-specific authorization model.

\subsection{Contributions}

\begin{enumerate}
\item \textbf{Formal Architecture:} We specify how off-chain encrypted storage combined with on-chain access control achieves confidentiality against curious storage providers, integrity verification, and patient-controlled authorization with cryptographic attribution.

\item \textbf{Reference Implementation:} Complete Ethereum smart contract handling record registration, permission granting through signed messages with explicit nonce management, time-bounded access with revocation, update/rotation capabilities, and comprehensive auditability through event logs.

\item \textbf{Healthcare Integration:} Integration with HL7 FHIR standards, showing how FHIR resources serve as plaintext while supporting de-identified data release.

\item \textbf{Performance Characterization:} Empirical evaluation showing gas costs, latency profiles, and scalability across Layer-1 and Layer-2 deployments.
\end{enumerate}

\section{Background}

\subsection{Threat Model}

We consider adversaries with varying capabilities:

\begin{enumerate}
\item \textbf{Storage Provider ($\mathcal{S}$):} Honest-but-curious cloud provider (IPFS, AWS S3) who stores encrypted records. $\mathcal{S}$ follows protocol but attempts to learn patient information from stored data.

\item \textbf{Network Adversary ($\mathcal{N}$):} Observes blockchain transactions and network traffic. Cannot break cryptographic primitives but can analyze patterns, timing, and metadata.

\item \textbf{Revoked Recipient ($\mathcal{R}$):} Previously authorized healthcare provider whose access was revoked. Possesses historical wrapped keys and may have cached plaintext from authorized period.

\item \textbf{Malicious Provider ($\mathcal{M}$):} Healthcare provider attempting unauthorized access or privilege escalation beyond granted permissions.
\end{enumerate}

We assume cryptographic primitives are secure: adversaries cannot break AES-256-GCM, ECIES on secp256k1 (using standard KDF/MAC/encoding from audited libraries), or forge ECDSA signatures.

\subsection{Cryptographic Building Blocks}

\textbf{Symmetric Encryption:} AES-256-GCM provides authenticated encryption with associated data (AEAD). Given key $K$, nonce $N$, plaintext $M$, and associated data $AD$: $\text{Enc}(K,N,M,AD) \rightarrow (C,T)$ where $C$ is ciphertext, $T$ is authentication tag.

\textbf{Nonce Requirements:} AES-GCM security critically depends on nonce uniqueness. Implementations MUST use cryptographically secure random number generators (CSPRNG) or counter-mode deterministic random bit generators (CTR-DRBG) to ensure uniqueness. Consider XChaCha20-Poly1305 if nonce management is operationally risky, as it provides a larger nonce space and better misuse-resistance.

\textbf{Public Key Encryption:} ECIES (Elliptic Curve Integrated Encryption Scheme) on secp256k1 curve provides IND-CCA2 secure public key encryption. We use the following configuration for standards compliance:
\begin{itemize}
\item KDF: ANSI X9.63 with SHA-256
\item DEM: AES-128-CTR
\item MAC: HMAC-SHA-256
\item Point encoding: Uncompressed (0x04 prefix)
\item Ephemeral key included in ciphertext
\end{itemize}

\textbf{Digital Signatures:} ECDSA signatures on secp256k1 provide authentication and non-repudiation. EIP-712 structured data signing prevents signature malleability across contexts.

\textbf{Cryptographic Hash Functions:} We use SHA-256 for content digests throughout the system. SHA-256 provides 256-bit collision resistance and is the standard for IPFS content addressing and general cryptographic applications. The digest $d$ is consistently defined as $d = \text{SHA-256}(C || T || N || AD)$ where $C$ is ciphertext, $T$ is authentication tag, $N$ is nonce, and $AD$ is associated data.

\subsection{Blockchain Infrastructure}

\textbf{Smart Contracts:} Ethereum smart contracts execute deterministically based on transaction inputs. Gas fees incentivize efficient code and prevent denial-of-service.

\textbf{Layer-2 Solutions:} Rollups (Optimistic/ZK) reduce costs by batching transactions. However, data availability charges often dominate L2 costs, particularly for calldata-heavy operations like storing wrapped keys.

\textbf{Events:} Smart contract events provide efficient, queryable logs. Events cost 375 gas base + 375 gas per topic + 8 gas/byte for data.

\subsection{Healthcare Standards}

\textbf{HL7 FHIR:} Fast Healthcare Interoperability Resources define standard formats for clinical data. Resources include Patient, Observation, Medication, Procedure, etc. We use FHIR R4.

\textbf{HIPAA:} Requires access controls, audit logs, and data encryption. Our architecture maps to HIPAA's administrative (identity-based access control with explicit authorization), physical (N/A for digital systems), and technical safeguards (encryption at rest and in transit).

\textbf{GDPR:} European privacy regulation granting data subject rights. Blockchain immutability creates tension with ``right to be forgotten''—we address through minimal on-chain data, treating blockchain entries as legally-required audit logs potentially exempt from erasure under Article 17(3)(b), and implementing data minimization strategies.

\section{System Architecture}

\subsection{Overview}

The system separates concerns across three layers:

\begin{enumerate}
\item \textbf{Storage Layer:} Distributed storage (IPFS) or cloud storage (S3) holds encrypted health records. Storage providers see only ciphertext.

\item \textbf{Blockchain Layer:} Ethereum smart contracts manage metadata, permissions, and audit trails. Each patient has a dedicated contract instance.

\item \textbf{Application Layer:} Client applications handle encryption/decryption, signature generation, and user interfaces.
\end{enumerate}

\begin{figure}[t]
\centering
\includegraphics[width=\linewidth]{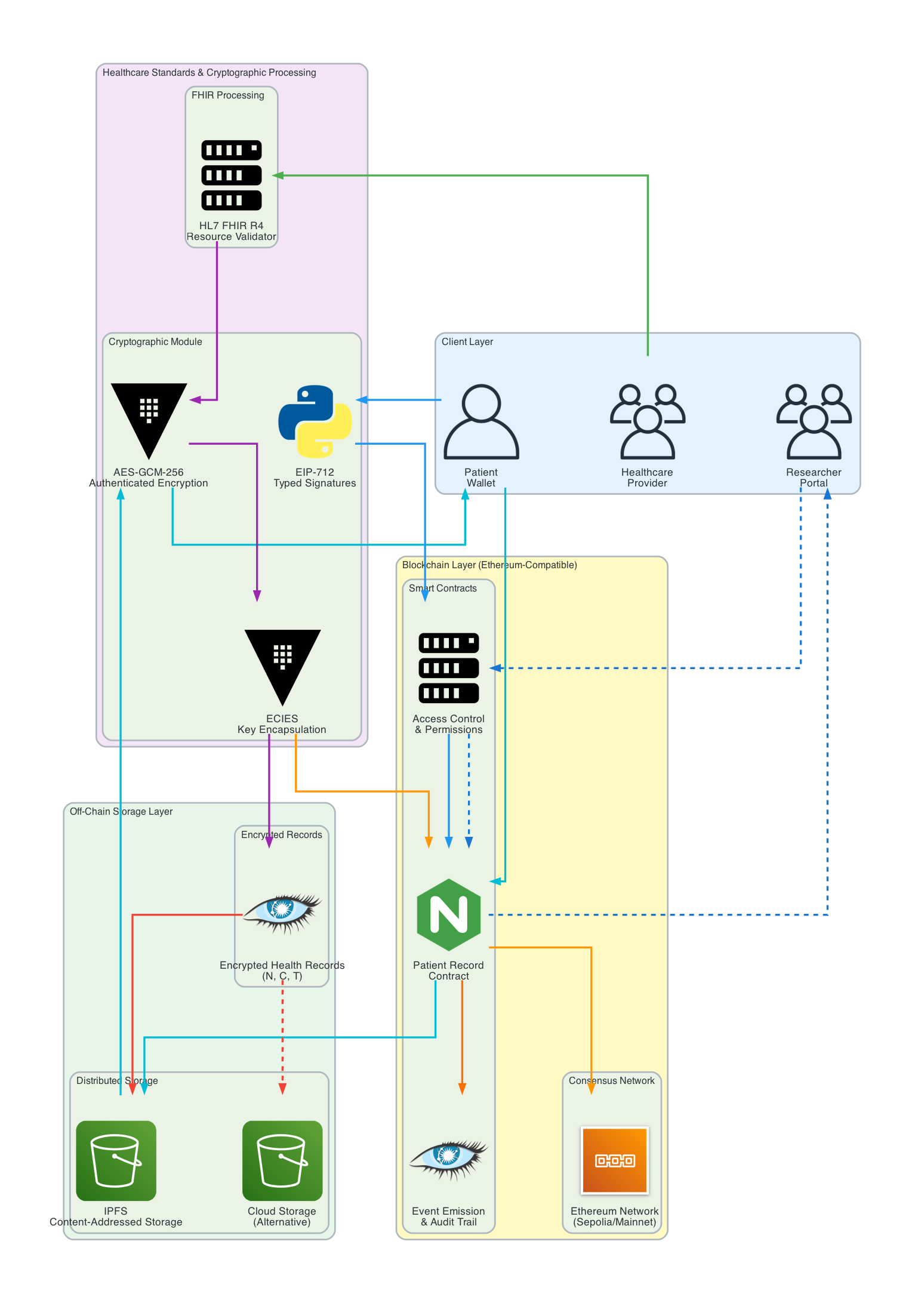}
\caption{System architecture showing separation between on-chain authorization and off-chain encrypted storage. Digest $d = \text{SHA-256}(C || T || N || AD)$.}
\label{fig:architecture}
\end{figure}

\subsection{Key Management Architecture}

\textbf{Key Registry Contract:} A separate registry contract maintains current encryption public keys for all participants. The invariant is that \texttt{getKey(user)} returns the latest non-revoked key and its version. Clients MUST fetch the key immediately prior to wrapped key computation to avoid time-of-check-time-of-use (TOCTOU) issues:

\begin{lstlisting}[caption={Key Registry Interface}, label={lst:keyregistry}]
interface IKeyRegistry {
    function registerKey(bytes memory publicKey) external;
    function rotateKey(bytes memory newPublicKey) external;
    function getKey(address user) 
        external view returns (bytes memory, uint256 version);
    function revokeKey() external;
    
    event KeyRegistered(address indexed user, bytes publicKey);
    event KeyRotated(address indexed user, 
        bytes newKey, uint256 version);
    event KeyRevoked(address indexed user);
}
\end{lstlisting}

\subsection{Data Model}

Each health record consists of:

\begin{itemize}
\item \textbf{Record ID ($rid$):} Unique identifier within patient's contract
\item \textbf{Plaintext ($M$):} FHIR resource bundle in JSON format
\item \textbf{Symmetric Key ($SymmK$):} AES-256 key for record encryption
\item \textbf{Ciphertext ($C$):} Encrypted record stored off-chain
\item \textbf{Storage Pointer ($ptr$):} IPFS CID or S3 URL
\item \textbf{Content Digest ($d$):} SHA-256 hash of complete ciphertext blob: $d = \text{SHA-256}(C || T || N || AD)$
\item \textbf{Wrapped Keys ($W$):} ECIES-encrypted $SymmK$ for authorized parties
\item \textbf{Permissions:} Time-bounded access grants with wrapped keys
\end{itemize}

\subsection{AEAD Metadata Considerations}

When using AEAD, the associated data (AD) parameter authenticates but does NOT encrypt additional context. To prevent metadata leakage:

\begin{enumerate}
\item \textbf{Minimal AD:} Use only non-sensitive, constant values (e.g., version number, fixed resource type identifier)
\item \textbf{Encrypted Metadata:} Include sensitive metadata (timestamps, detailed resource types) within the encrypted payload $M$ itself
\item \textbf{Constant Format:} Ensure AD format doesn't vary in ways that leak information through length or structure
\end{enumerate}

\section{Protocol Workflows}

\subsection{Workflow 1: Record Creation}

Patient creates new health record:

\begin{enumerate}
\item \textbf{Prepare Plaintext:} Construct FHIR bundle $M$ containing clinical data.

\item \textbf{Generate Symmetric Key:} Generate random AES-256 key: $SymmK \leftarrow \{0,1\}^{256}$ using CSPRNG.

\item \textbf{Encrypt Record:} Using AES-256-GCM with CSPRNG-generated nonce:
   $$(C, T) \leftarrow \text{AES-GCM-Enc}(SymmK, N, M, AD_{minimal})$$
   where $AD_{minimal}$ contains only non-sensitive version identifier.

\item \textbf{Upload to Storage:} Store $(C, T, N, AD_{minimal})$ to IPFS/S3. Receive storage pointer $ptr$.

\item \textbf{Compute Digest:} $d \leftarrow \text{SHA-256}(C || T || N || AD_{minimal})$.

\item \textbf{Wrap Key for Owner:} Using patient's public key from registry:
   $$W_{owner} \leftarrow \text{ECIES-Enc}(PK_P, SymmK)$$

\item \textbf{Register On-Chain:} Call \texttt{addRecord}$(ptr, d, W_{owner})$. Contract stores metadata, assigns $rid$, emits \texttt{RecordAdded}$(rid, d, ptr)$.
\end{enumerate}

\subsection{Workflow 2: Permission Grant}

Patient grants time-bounded access using EIP-712 signatures. Note that patients MUST generate unique nonces for each grant (e.g., 256-bit random values) to enable parallel permission grants. Figure~\ref{fig:permission-grant} illustrates this workflow:

\begin{enumerate}
\item \textbf{Retrieve Recipient Key:} Query key registry for recipient's current public key $PK_R$ immediately before wrapping.

\item \textbf{Wrap Symmetric Key:} $W_R \leftarrow \text{ECIES-Enc}(PK_R, SymmK)$.

\item \textbf{Generate Unique Nonce:} Patient generates unique nonce (256-bit random recommended).

\item \textbf{Prepare Typed Data:} Construct EIP-712 message with explicit nonce:
\begin{verbatim}
{
  recordId: rid,
  grantee: addr_R,
  expiration: timestamp,
  wrappedKey: W_R,
  nonce: uniqueRandomNonce
}
\end{verbatim}

\item \textbf{Sign Message:} $\sigma \leftarrow \text{ECDSA-Sign}(SK_P, \text{TypedDataHash}(message))$.

\item \textbf{Recipient Submits:} Recipient calls \texttt{grantPermissionBySig}$(rid, expiration, W_R, nonce, \sigma)$.

\item \textbf{Contract Verification:} Contract verifies signature, checks nonce hasn't been used, stores permission, marks nonce as consumed, emits \texttt{PermissionGranted}.
\end{enumerate}

\begin{figure}[t]
\centering
\includegraphics[width=\linewidth]{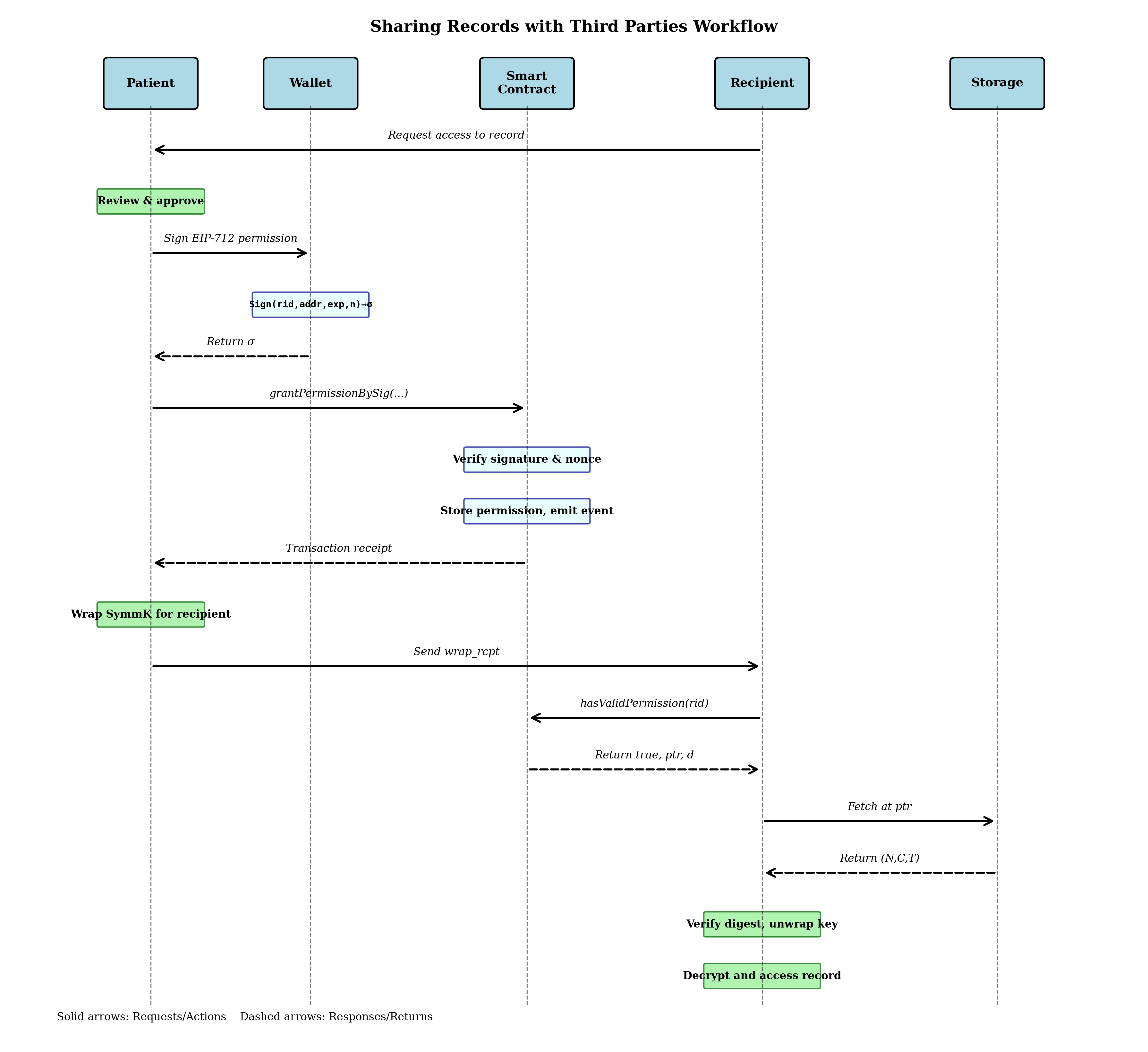}
\caption{Permission grant workflow using EIP-712 signed messages with explicit nonce management. Each nonce can be used only once.}
\label{fig:permission-grant}
\end{figure}

\subsection{Workflow 3: Record Access}

Authorized recipient retrieves and decrypts record as shown in Figure~\ref{fig:patient-access}. Note that the contract optionally gates metadata for UX consistency; confidentiality relies solely on encryption:

\begin{enumerate}
\item \textbf{Request Metadata:} Call \texttt{getRecordMetadata}$(rid)$ which returns $(ptr, d)$ if authorized. The gating is for user experience; the same data is available in public events. Patient additionally receives $W_{owner}$ through \texttt{getOwnerWrappedKey}$(rid)$.

\item \textbf{Retrieve Ciphertext:} Fetch $(C, T, N, AD)$ from storage using $ptr$.

\item \textbf{Verify Integrity:} Compute $d' = \text{SHA-256}(C || T || N || AD)$. Verify $d' = d$.

\item \textbf{Unwrap Key:} Decrypt wrapped key using recipient's private key:
   $$SymmK \leftarrow \text{ECIES-Dec}(SK_R, W_R)$$

\item \textbf{Decrypt Record:} $M \leftarrow \text{AES-GCM-Dec}(SymmK, N, C, T, AD)$.

\item \textbf{Optional: Log Access:} Call \texttt{logAccess}$(rid, \text{SHA-256}(accessDetails))$ to create on-chain access receipt (not just authorization).
\end{enumerate}

\begin{figure}[t]
\centering
\includegraphics[width=\linewidth]{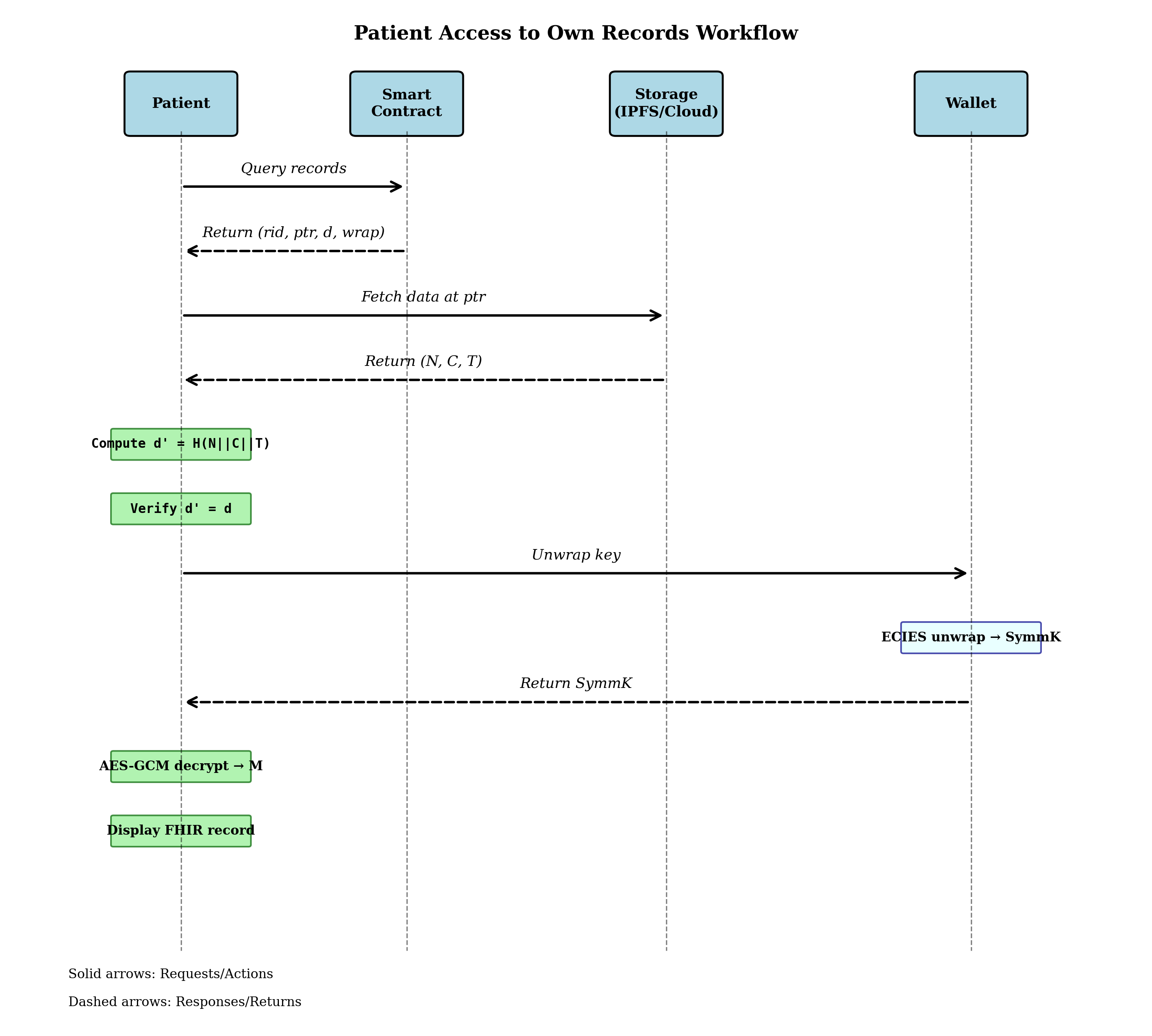}
\caption{Record access workflow. Contract gates metadata for UX consistency (data also in events). Digest verification: $d = \text{SHA-256}(C || T || N || AD)$.}
\label{fig:patient-access}
\end{figure}

\subsection{Workflow 4: Permission Revocation}

Patient revokes access:

\begin{enumerate}
\item \textbf{Submit Revocation:} Patient calls \texttt{revokePermission}$(rid, addr_R)$. Contract sets $\text{permissions}[rid][addr_R].\text{revoked} = true$ and emits \texttt{PermissionRevoked}$(rid, addr_R)$.

\item \textbf{Future Access Blocked:} Subsequent \texttt{getRecordMetadata}$(rid)$ calls by $addr_R$ fail permission check.
\end{enumerate}

\textbf{Limitation:} Revocation prevents future access but cannot retract already-decrypted plaintext. If recipient downloaded $M$ before revocation, they retain that data. This is fundamental to cryptographic access control and must be clearly communicated to patients.

\subsection{Workflow 5: Record Update/Key Rotation}

For stronger guarantees when trust relationships end or to update record content:

\begin{enumerate}
\item \textbf{Generate New Key:} $SymmK' \leftarrow \{0,1\}^{256}$.

\item \textbf{Re-encrypt Record:} Encrypt updated content $M'$ with $SymmK'$, upload to storage, get new $ptr'$.

\item \textbf{Compute New Digest:} $d' = \text{SHA-256}(C' || T' || N' || AD')$.

\item \textbf{Update On-Chain:} Call \texttt{updateRecord}$(rid, ptr', d', W'_{owner})$. 

\item \textbf{Invalidate Old Version:} Previous $(ptr, d, SymmK)$ become obsolete. Revoked recipients cannot access new version.

\item \textbf{Emit Event:} Contract emits \texttt{RecordUpdated}$(rid, d', ptr')$ for audit trail.
\end{enumerate}

\section{Smart Contract Implementation}

\subsection{Design Rationale}

The smart contract serves three roles: (1) \textbf{Metadata Registry}—storing storage pointers and content digests; (2) \textbf{Authorization Engine}—verifying permissions before metadata release; (3) \textbf{Audit Log}—emitting events for all operations.

We deploy \textbf{one contract per patient} using a simple authorization model without token standards. This provides:
\begin{itemize}
\item Strong isolation between patients
\item Simplified permission model (no cross-patient checks)
\item Clear ownership semantics
\item Independent upgrade paths per patient
\end{itemize}

While this increases deployment costs, it eliminates shared-state vulnerabilities and aligns with patient sovereignty.

\subsection{Core Data Structures}

\begin{lstlisting}[caption={Smart Contract Data Structures}, label={lst:structs}]
contract PatientHealthRecords is EIP712 {
    address public immutable patient;
    uint256 private _recordCounter;
    
    struct RecordMetadata {
        string storagePointer;
        bytes32 contentDigest;
        bytes wrappedKeyOwner;
        uint64 createdAt;
        uint64 updatedAt;
    }
    
    struct Permission {
        uint64 expiration;
        bool revoked;
        bytes wrappedKey;
    }
    
    mapping(uint256 => RecordMetadata) private _records;
    mapping(uint256 => mapping(address => Permission)) 
        public permissions;
    mapping(bytes32 => bool) public usedNonces;
    
    modifier onlyPatient() {
        require(msg.sender == patient, "Only patient");
        _;
    }
    
    modifier validRecordId(uint256 rid) {
        require(rid > 0 && rid <= _recordCounter, 
            "Invalid record ID");
        _;
    }
    
    // Event Declarations
    event RecordAdded(uint256 indexed rid, bytes32 digest, string ptr);
    event RecordUpdated(uint256 indexed rid, bytes32 digest, string ptr);
    event PermissionGranted(uint256 indexed rid, 
        address indexed grantee, uint64 expiration);
    event PermissionRevoked(uint256 indexed rid, 
        address indexed grantee);
    event EmergencyAccessGranted(bytes32 indexed grantId,
        uint256 indexed rid, address physician1, address physician2,
        uint8 justificationCode, uint64 expiration, uint64 requestTime);
    event EmergencyAccessConfirmed(bytes32 indexed grantId,
        uint256 indexed rid, address physician, 
        bytes32 justificationHash);
    event AccessLogged(uint256 indexed rid, 
        address indexed accessor, bytes32 detailsHash);
}
\end{lstlisting}

\subsection{Record Registration}

\begin{lstlisting}[caption={Record Registration Function}, label={lst:addrec}]
function addRecord(
    string memory ptr,
    bytes32 digest,
    bytes memory wrappedKey
) external onlyPatient returns (uint256) {
    _recordCounter++;
    uint256 rid = _recordCounter;
    
    _records[rid] = RecordMetadata({
        storagePointer: ptr,
        contentDigest: digest,
        wrappedKeyOwner: wrappedKey,
        createdAt: uint64(block.timestamp),
        updatedAt: uint64(block.timestamp)
    });
    
    emit RecordAdded(rid, digest, ptr);
    return rid;
}
\end{lstlisting}

Gas cost: $\sim$180,000 gas first record (cold storage), $\sim$165,000 gas subsequent records.

\subsection{Signature-Based Permission Grant with Explicit Nonce}

\begin{lstlisting}[caption={Permission Grant with Explicit Nonce Management}, label={lst:grant}, float=t]
function grantPermissionBySig(
    uint256 rid,
    uint64 expiration,
    bytes memory wrappedKey,
    uint256 nonce,
    bytes memory signature
) external validRecordId(rid) {
    // Construct nonce hash to prevent reuse
    bytes32 nonceHash = keccak256(
        abi.encodePacked(patient, nonce)
    );
    require(!usedNonces[nonceHash], "Nonce already used");
    
    // Recover signer with provided nonce
    address signer = _recoverSigner(
        rid, msg.sender, expiration,
        wrappedKey, nonce, signature
    );
    require(signer == patient, "Invalid signature");
    
    // Check expiration is future
    require(expiration > block.timestamp, 
        "Expiration must be future");
    
    // Mark nonce as used
    usedNonces[nonceHash] = true;
    
    // Store permission
    permissions[rid][msg.sender] = Permission({
        expiration: expiration,
        revoked: false,
        wrappedKey: wrappedKey
    });
    
    emit PermissionGranted(rid, msg.sender, expiration);
}

function _recoverSigner(
    uint256 rid, address grantee, uint64 expiration,
    bytes memory wk, uint256 nonce, bytes memory sig
) internal view returns (address) {
    bytes32 structHash = keccak256(abi.encode(
        PERMISSION_TYPEHASH,
        rid, grantee, expiration,
        keccak256(wk), nonce
    ));
    bytes32 hash = _hashTypedDataV4(structHash);
    return ECDSA.recover(hash, sig);
}
\end{lstlisting}

Gas cost: $\sim$78,000 gas. The ecrecover precompile itself costs $\sim$3,000 gas; storage operations and calldata processing account for the remainder.

\subsection{Permission Verification and Metadata Access}

\begin{lstlisting}[caption={Permission Check with Owner Key Retrieval}, label={lst:access}, float=t]
function hasValidPermission(uint256 rid)
    public view validRecordId(rid)
    returns (bool) {
    // Patient always has access
    if (msg.sender == patient) return true;
    
    Permission memory p = permissions[rid][msg.sender];
    return !p.revoked && 
           p.expiration > 0 &&
           p.expiration > block.timestamp; // Strict inequality
}

function getRecordMetadata(uint256 rid)
    external view validRecordId(rid)
    returns (string memory ptr, bytes32 digest) {
    require(hasValidPermission(rid),
        "Not authorized");
    
    RecordMetadata memory rec = _records[rid];
    ptr = rec.storagePointer;
    digest = rec.contentDigest;
    // Note: This gating is for UX; data is public in events
}

// New function for owner to retrieve their wrapped key
function getOwnerWrappedKey(uint256 rid)
    external view onlyPatient validRecordId(rid)
    returns (bytes memory) {
    return _records[rid].wrappedKeyOwner;
}
\end{lstlisting}

\subsection{Record Update and Key Rotation}

\begin{lstlisting}[caption={Record Update Function}, label={lst:update}, float=t]
function updateRecord(
    uint256 rid,
    string memory newPtr,
    bytes32 newDigest,
    bytes memory newOwnerWrappedKey
) external onlyPatient validRecordId(rid) {
    RecordMetadata storage rec = _records[rid];
    
    rec.storagePointer = newPtr;
    rec.contentDigest = newDigest;
    rec.wrappedKeyOwner = newOwnerWrappedKey;
    rec.updatedAt = uint64(block.timestamp);
    
    emit RecordUpdated(rid, newDigest, newPtr);
}
\end{lstlisting}

\subsection{Permission Revocation}

\begin{lstlisting}[caption={Permission Revocation}, label={lst:revoke}, float=t]
function revokePermission(
    uint256 rid, address grantee
) external onlyPatient validRecordId(rid) {
    require(permissions[rid][grantee].expiration > 0,
        "No permission to revoke");
    
    permissions[rid][grantee].revoked = true;
    emit PermissionRevoked(rid, grantee);
}
\end{lstlisting}

Gas cost: $\sim$31,000 gas.

\subsection{Emergency Access Pattern}

When patients are incapacitated and cannot grant permissions, emergency access is critical. The two-physician multisignature pattern ensures medical necessity while maintaining accountability. The wrapped keys for emergency physicians are generated by an institutional guardian service (HSM-backed) that unwraps the patient's owner key (or uses a pre-established envelope key) and re-wraps for each authorized physician—see §9.2 for institutional key management details. The contract merely anchors authorization and audit.

\begin{lstlisting}[caption={Corrected Emergency Access Implementation}, label={lst:emergency}]
mapping(address => bool) public emergencyPhysicians;
mapping(bytes32 => EmergencyGrant) public emergencyGrants;

struct EmergencyRequest {
    uint256 rid;
    uint8 justificationCode;
    uint64 requestTime;
    uint64 maxSkewSeconds;
}

struct EmergencyGrant {
    uint256 recordId;
    address physician1;
    address physician2;
    uint64 expiration;
    bool confirmed;
}

function emergencyGrantAccess(
    uint256 rid,
    address physician2,
    uint8 justificationCode,
    uint64 requestTime,
    uint64 maxSkewSeconds,
    bytes memory wrappedKey1,
    bytes memory wrappedKey2,
    bytes memory signature1,
    bytes memory signature2
) external validRecordId(rid) {
    require(emergencyPhysicians[msg.sender], 
        "Not emergency physician");
    require(emergencyPhysicians[physician2], 
        "Not emergency physician");
    require(msg.sender != physician2, 
        "Different physicians required");
    
    // Check time skew tolerance
    uint256 timeDiff = (block.timestamp > requestTime) ? 
        block.timestamp - uint256(requestTime) : 
        uint256(requestTime) - block.timestamp;
    require(timeDiff <= uint256(maxSkewSeconds), 
        "Request time outside tolerance");
    
    // Verify both signatures over EIP-712 struct
    bytes32 requestHash = _hashTypedDataV4(
        keccak256(abi.encode(
            EMERGENCY_REQUEST_TYPEHASH,
            rid,
            justificationCode,
            requestTime,
            maxSkewSeconds
        ))
    );
    
    address signer1 = ECDSA.recover(requestHash, signature1);
    address signer2 = ECDSA.recover(requestHash, signature2);
    require(signer1 == msg.sender && signer2 == physician2, 
        "Invalid signatures");
    
    // Create 2-hour emergency grant
    uint64 emergencyExpiration = uint64(block.timestamp + 2 hours);
    
    // Store wrapped keys for both physicians
    permissions[rid][msg.sender] = Permission({
        expiration: emergencyExpiration,
        revoked: false,
        wrappedKey: wrappedKey1
    });
    
    permissions[rid][physician2] = Permission({
        expiration: emergencyExpiration,
        revoked: false,
        wrappedKey: wrappedKey2
    });
    
    // Compute deterministic grant ID
    bytes32 grantId = keccak256(abi.encode(
        rid, requestTime, msg.sender, physician2
    ));
    
    // Record emergency grant for audit
    emergencyGrants[grantId] = EmergencyGrant({
        recordId: rid,
        physician1: msg.sender,
        physician2: physician2,
        expiration: emergencyExpiration,
        confirmed: false
    });
    
    emit EmergencyAccessGranted(
        grantId,  // Include grantId for easy confirmation
        rid, 
        msg.sender, 
        physician2, 
        justificationCode, 
        emergencyExpiration,
        requestTime
    );
}

function confirmEmergencyAccess(
    bytes32 grantId,
    bytes32 justificationHash
) external {
    EmergencyGrant storage grant = emergencyGrants[grantId];
    require(msg.sender == grant.physician1 || 
            msg.sender == grant.physician2, 
            "Not authorized physician");
    require(!grant.confirmed, "Already confirmed");
    require(block.timestamp <= grant.expiration + 24 hours,
            "Confirmation window expired");
    
    grant.confirmed = true;
    emit EmergencyAccessConfirmed(
        grantId,  // Include grantId for tracking
        grant.recordId, 
        msg.sender, 
        justificationHash
    );
}
\end{lstlisting}

\subsection{Optional Access Logging}

\begin{lstlisting}[caption={Optional Read Receipt Logging}, label={lst:logging}]
event AccessLogged(uint256 indexed rid, 
    address indexed accessor, bytes32 detailsHash);

function logAccess(uint256 rid, bytes32 detailsHash) 
    external {
    require(hasValidPermission(rid), "Not authorized");
    emit AccessLogged(rid, msg.sender, detailsHash);
}
\end{lstlisting}

\subsection{Security Properties}

\textbf{Read Gating:} \texttt{getRecordMetadata} enforces authorization via \texttt{hasValidPermission} for UX consistency. The same metadata is available in public events; confidentiality relies entirely on encryption, not on gating. The \texttt{permissions} mapping is public for transparency and indexing; confidentiality relies solely on encryption—exposing \texttt{wrappedKey} ciphertexts does not endanger plaintext.

\textbf{Complete Owner Access:} Patient can always retrieve their wrapped key via \texttt{getOwnerWrappedKey}, ensuring they never lose access to their own records.

\textbf{Transparency Trade-off:} Storage pointers $ptr$ and digests $d$ are public in events and contract storage. This is acceptable because pointers reference encrypted content. Without wrapped keys, adversaries obtain only ciphertext.

\textbf{Replay Protection:} Explicit nonce management prevents signature replay. Each nonce can be used exactly once. Patients must generate unique nonces (e.g., 256-bit random values) for each grant to enable parallel permission grants without ordering hazards. Clients \textbf{MUST} generate a fresh 256-bit random nonce per grant and persist it until on-chain confirmation to avoid accidental reuse.

\textbf{Time-Bounded Access:} All permissions have expiration timestamps. Expired permissions fail \texttt{hasValidPermission} checks automatically, using strict inequality ($>$ rather than $\geq$) for clear expiration semantics.

\section{Security Analysis}

\subsection{Confidentiality Against Storage Providers}

\begin{theorem}[Storage Provider Confidentiality]
Assuming AES-GCM provides IND-CCA2 security and ECIES (with specified parameters) provides IND-CCA2 security, no honest-but-curious storage provider $\mathcal{S}$ can distinguish encrypted health records from random strings with non-negligible advantage.
\end{theorem}

\begin{proof}[Proof Sketch]
By contradiction. Suppose adversary $\mathcal{S}$ has non-negligible advantage $\epsilon$ in distinguishing encrypted records. Storage contains $(C, T, N, AD_{minimal})$ where $(C, T)$ are outputs of $\text{AES-GCM-Enc}(SymmK, N, M, AD_{minimal})$ and $AD_{minimal}$ contains only non-sensitive version identifiers. By IND-CCA2 security of AES-GCM, $(C, T)$ is computationally indistinguishable from random strings without $SymmK$. Since $\mathcal{S}$ cannot obtain $SymmK$ (wrapped keys $W$ are ECIES ciphertexts using specified KDF/MAC, also indistinguishable from random without private keys), $\mathcal{S}$ cannot distinguish $(C, T)$ from random. This contradicts the assumption of advantage $\epsilon > 0$. $\square$
\end{proof}

\subsection{Integrity Verification}

\begin{theorem}[Tamper Detection]
Assuming SHA-256 is collision-resistant, any modification to stored records is detected with probability $\geq 1 - 2^{-128}$.
\end{theorem}

\begin{proof}[Proof Sketch]
On-chain digest $d = \text{SHA-256}(C || T || N || AD)$ commits to encrypted blob. To pass verification, adversary must produce $(C', T', N', AD')$ where $\text{SHA-256}(C' || T' || N' || AD') = d$ with $(C', T') \neq (C, T)$. This requires finding collision in SHA-256. With 256-bit output and collision resistance, success probability is $\leq 2^{-128}$ (birthday bound). $\square$
\end{proof}

\textbf{Authentication Tag:} AES-GCM's tag $T$ provides additional integrity protection. Even if adversary finds hash collision, forging valid tag without $SymmK$ has negligible probability (AES-GCM provides 128-bit authentication security).

\subsection{Authorization Authenticity}

\begin{theorem}[Cryptographically Attributable Authorization]
Assuming ECDSA over secp256k1 provides existential unforgeability under chosen message attacks (EUF-CMA), no adversary without patient's private key $SK_P$ can forge valid permission signatures with non-negligible probability.
\end{theorem}

\begin{proof}[Proof Sketch]
Suppose adversary $\mathcal{A}$ forges signature $\sigma'$ for message $(rid, addr_R, expiration, W_R, nonce)$ that passes verification. By EUF-CMA security of ECDSA, this occurs with negligible probability without $SK_P$. EIP-712 domain separator binds signature to specific contract and chain, preventing cross-contract/cross-chain replay. Unique nonces prevent same-message replay. $\square$
\end{proof}

\subsection{Replay Attack Prevention}

Nonces provide replay protection:
\begin{property}[Unique Nonce Consumption]
For each patient, each nonce can be used exactly once. Signature $\sigma$ for $(rid, addr_R, expiration, W_R, nonce)$ is valid only if nonce has not been previously used. After successful grant, nonce is marked as consumed, invalidating all future signatures using the same nonce.
\end{property}

\textbf{Time-Based Replay:} Expiration timestamps prevent long-term replay. Even if adversary captures signature, using it after expiration fails \texttt{require(expiration > block.timestamp)} check.

\subsection{Auditability}

\begin{property}[Authorization Audit Trail Completeness]
All authorization-changing operations emit events: \texttt{RecordAdded}, \texttt{PermissionGranted}, \texttt{PermissionRevoked}, \texttt{RecordUpdated}, \texttt{EmergencyAccessGranted}. Events are permanently stored in blockchain logs, queryable by any observer. The system provides complete authorization history (who was granted access), not complete access history (who actually retrieved/viewed records) unless optional \texttt{logAccess} is used.
\end{property}

Blockchain immutability ensures events cannot be deleted or modified after confirmation. Patients, regulators, or auditors can reconstruct complete authorization history by filtering events for specific records or addresses.

\subsection{Threat Analysis Summary}

Table~\ref{tab:threats} summarizes threat coverage.

\begin{table}[t]
\centering
\caption{Threat Coverage}
\label{tab:threats}
\small
\begin{tabular}{@{}lll@{}}
\toprule
\textbf{Threat} & \textbf{Defense} & \textbf{Residual Risk} \\
\midrule
Storage snooping & Encryption & None \\
Data tampering & Digest verification & None \\
Unauthorized access & On-chain authz & None \\
Permission forgery & EIP-712 signatures & None \\
Replay attacks & Unique nonces & None \\
Audit log tampering & Blockchain immutability & None \\
Patient key theft & --- & High \\
Malicious patient & --- & Inherent \\
Storage unavailability & Redundancy & Low \\
DoS on blockchain & Fees deter, multi-L2 & Medium \\
\bottomrule
\end{tabular}
\end{table}

\section{Performance Evaluation}

\subsection{Experimental Setup}

\textbf{Blockchain:} Ethereum Sepolia testnet, September-October 2024. Transactions via Web3.js v4.2.1.

\textbf{Storage:} IPFS (go-ipfs v0.18) on 8-core, 16 GB RAM server. AWS S3 for comparison. Client in Dhaka, Bangladesh; IPFS nodes in Singapore; S3 in us-east-1.

\textbf{Client:} Intel Core i7-1165G7 @ 2.80GHz, 16 GB RAM, Ubuntu 22.04. Web Crypto API for AES-GCM, eth-crypto for ECIES/ECDSA.

\textbf{Workload:} Synthea v3.2.0 generated FHIR R4 resources. Record sizes: 1 KB (observations) to 10 MB (imaging reports). 50 trials per measurement for mean and 95th percentile.

\subsection{Cryptographic Operations}

Table~\ref{tab:crypto-ops} shows operation latencies.

\begin{table}[t]
\centering
\caption{Cryptographic Operation Latency}
\label{tab:crypto-ops}
\small
\begin{tabular}{@{}lrr@{}}
\toprule
\textbf{Operation} & \textbf{Mean (ms)} & \textbf{95th \% (ms)} \\
\midrule
AES-GCM Enc (1 KB) & 0.42 & 0.58 \\
AES-GCM Enc (100 KB) & 2.1 & 2.7 \\
AES-GCM Enc (1 MB) & 18.3 & 22.1 \\
AES-GCM Enc (10 MB) & 181.5 & 205.3 \\
AES-GCM Dec (1 MB) & 16.8 & 20.5 \\
ECIES Key Wrap & 3.2 & 4.1 \\
ECIES Key Unwrap & 3.5 & 4.3 \\
ECDSA Sign (EIP-712) & 2.8 & 3.6 \\
ECDSA Verify & 3.1 & 3.9 \\
SHA-256 (1 MB) & 12.1 & 14.8 \\
\bottomrule
\end{tabular}
\end{table}

AES-GCM achieves $\sim$55 MB/s throughput, linear in plaintext size. Hardware AES acceleration (AES-NI) provides these speeds on modern processors. Public-key operations (ECIES, ECDSA) take 3-4 ms regardless of record size—operate on fixed-size keys/hashes. For typical sharing (one signature, one key wrap), total cryptographic overhead <10 ms, negligible vs network latency.

\subsection{On-Chain Gas Costs}

Table~\ref{tab:gas-detailed} reports gas consumption.

\begin{table}[t]
\centering
\caption{Smart Contract Gas Consumption}
\label{tab:gas-detailed}
\small
\begin{tabular}{@{}lr@{}}
\toprule
\textbf{Operation} & \textbf{Gas} \\
\midrule
\multicolumn{2}{l}{\textit{Ethereum Mainnet (L1)}} \\
Contract Deployment & 2,341,829 \\
addRecord (first) & 183,742 \\
addRecord (subsequent) & 166,542 \\
grantPermissionBySig & 78,331 \\
revokePermission & 31,204 \\
updateRecord & 45,123 \\
emergencyGrantAccess & 156,432 \\
confirmEmergencyAccess & 35,211 \\
\midrule
\multicolumn{2}{l}{\textit{Layer-2 (Arbitrum One)}} \\
addRecord & 14,392 \\
grantPermissionBySig & 6,127 \\
\midrule
\multicolumn{2}{l}{\textit{Layer-2 (zkSync Era)}} \\
addRecord & 11,243 \\
grantPermissionBySig & 5,894 \\
\bottomrule
\end{tabular}
\end{table}

Layer-2 solutions reduce gas consumption 10-13$\times$ through batching and off-chain computation. However, data availability charges for posting calldata to L1 can dominate actual costs during high congestion periods. Fees on rollups are dominated by L1 data-availability; published 'gas used' reductions (Table 3) do not directly translate linearly to USD costs.

\textbf{Cost Breakdown:} \texttt{grantPermissionBySig} costs: signature verification ($\sim$3,000 gas for ecrecover precompile), storage writes ($\sim$40,000 gas for new permission), state updates ($\sim$20,000 gas), events ($\sim$10,000 gas), execution overhead ($\sim$5,000 gas).

\textbf{Economic Viability:} At \$3,000/ETH and 30 gwei gas price, permission grant costs $\sim$\$7 on L1, $\sim$\$0.54 on L2. Healthcare institutions or insurance can subsidize L1 costs; L2 enables direct patient payment models with substantially lower costs.

\subsection{End-to-End Access Latency}

Table~\ref{tab:e2e-latency} shows total record access time.

\begin{table}[t]
\centering
\caption{End-to-End Access Latency (1 MB Records)}
\label{tab:e2e-latency}
\small
\begin{tabular}{@{}lrr@{}}
\toprule
\textbf{Component} & \textbf{Mean (ms)} & \textbf{95th \% (ms)} \\
\midrule
Blockchain query & 245 & 312 \\
IPFS retrieval & 1,087 & 1,523 \\
Integrity verification (SHA-256) & 12 & 15 \\
Key unwrapping (ECIES) & 4 & 5 \\
AES-GCM decryption & 17 & 21 \\
\midrule
\textbf{Total (IPFS)} & \textbf{1,365} & \textbf{1,876} \\
\midrule
S3 retrieval & 423 & 589 \\
\textbf{Total (S3)} & \textbf{701} & \textbf{942} \\
\bottomrule
\end{tabular}
\end{table}

Storage retrieval dominates latency. IPFS exhibits higher variance due to distributed nature—retrieval from distant/slow peers. S3 provides consistent performance through CDN. Cryptographic operations contribute <5\% total latency. The abstract reports mean values (0.7--1.4 s) for typical performance expectations.

\textbf{Scalability:} For 10 MB records, total latency increases to $\sim$3.5 s (IPFS) or $\sim$2.1 s (S3). Encryption/decryption scale linearly but remain small fraction. For typical clinical documents (10-500 KB), latency stays $<$1 s.

\subsection{Storage Overhead}

Encryption overhead: AES-GCM adds 12 bytes (nonce) + 16 bytes (tag) = 28 bytes per record. For 1 MB plaintext: overhead 0.003\%. Content-addressing (IPFS CID) adds 32-byte identifier. Total storage penalty negligible.

On-chain storage per record: 32 bytes (digest) + 50 bytes (storage pointer string) + 110-150 bytes (wrapped key, implementation-dependent with ECIES ephemeral public key + IV + MAC + ciphertext) $\approx$ 192-232 bytes. Persistent storage on Ethereum is charged per 32-byte storage slot (20,000 gas for first-write). Dynamic types (strings/bytes) span multiple slots. In practice, \texttt{addRecord} costs $\sim$166-184k gas on L1 (Table 3), dominated by storage writes and event emission; this measured figure is more representative than per-byte estimates.

\section{Privacy and Regulatory Compliance}

\subsection{Metadata Privacy}

\textbf{Public Information:} Storage pointers $ptr$ and content digests $d$ are public on-chain, visible in events and contract storage. This transparency is \emph{by design}—pointers reference exclusively \emph{encrypted} content. Without wrapped keys, adversaries obtain only ciphertext.

\textbf{Metadata Leakage:} On-chain data reveals: (1) Which addresses participate in health data sharing; (2) How many records each patient has; (3) When records are created/accessed; (4) Which recipients have permissions.

\textbf{Mitigation:} Recipients concerned about linkability should use fresh addresses per patient via HD wallet derivation (BIP-32/BIP-44). Patients can use mixing services or privacy-preserving layer-2s (zkSync) to obscure transaction origins. However, complete metadata privacy contradicts auditability—trade-off between transparency and privacy must be balanced per deployment requirements.

\subsection{HIPAA Compliance}

Health Insurance Portability and Accountability Act (HIPAA) mandates safeguards for protected health information (PHI). Our system addresses HIPAA requirements:

\textbf{Administrative Safeguards:} Patients define access policies through cryptographic permissions. Audit logs track all authorization events. Risk analysis identifies vulnerabilities (key management, storage availability).

\textbf{Physical Safeguards:} Encrypted storage prevents PHI exposure during theft/breach. Hardware wallets protect private keys.

\textbf{Technical Safeguards:} Authentication (ECDSA signatures), encryption (AES-256), integrity (SHA-256), audit trails (blockchain events for authorization history), transmission security (TLS for off-chain communication).

\textbf{Access Control:} Identity-based access through permission grants. Patient controls who accesses what, when. Emergency access mechanisms (Section 6.8) balance safety with privacy.

\textbf{Audit Logs (§164.312(b)):} Blockchain events create tamper-proof audit trail of all authorization events. The system provides complete authorization history (who was granted access) though not access history (who actually retrieved records) unless optional \texttt{logAccess} is used.

\subsection{GDPR Compliance}

General Data Protection Regulation (GDPR) grants data subjects rights over personal data. Our architecture supports GDPR requirements:

\textbf{Right to Access:} Patients always have permission to their records via \texttt{getOwnerWrappedKey}. Can query blockchain for complete metadata and retrieve encrypted data anytime.

\textbf{Right to Portability:} FHIR format enables standards-based data export. Patients can download encrypted records, decrypt with their keys, and transfer to other systems.

\textbf{Right to Erasure:} Complex due to blockchain immutability. We implement a three-tier approach:
\begin{enumerate}
\item Off-chain data can be deleted from storage
\item On-chain pointers can be nullified through contract updates
\item Permission/audit events remain as legally-required logs
\end{enumerate}

Organizations should establish legal basis treating blockchain entries as audit logs potentially exempt from erasure under regulatory compliance requirements (GDPR Article 17(3)(b)).

\textbf{Data Minimization:} Only encrypted pointers and hashes stored on-chain. No protected health information appears in blockchain. Metadata in associated data minimized to version identifiers.

\textbf{Consent Management:} EIP-712 signatures provide explicit, informed, unambiguous consent for each data sharing event with cryptographic proof.

\subsection{De-Identified Data for Research}

Healthcare research requires large datasets while protecting privacy. We support de-identification pipelines creating research-safe datasets.

\textbf{De-Identification Methods:} We implement Safe Harbor method (HIPAA) removing 18 identifier categories: names, addresses (except ZIP regions), dates (except year), phone numbers, emails, SSNs, medical record numbers, account numbers, certificate numbers, vehicle identifiers, device identifiers, URLs, IP addresses, biometric identifiers, full-face photos, and other unique identifiers.

\textbf{Pipeline Architecture:} Patients select records for research contribution. De-identification service: (1) Decrypts records using patient-authorized temporary key; (2) Applies transformations (identifier removal, quasi-identifier generalization, date shifting); (3) Re-encrypts with research consortium key; (4) Uploads to research database with new permissions.

\textbf{Performance:} Evaluation on 1,000 Synthea FHIR bundles (mean 127 KB): de-identification takes mean 47.3 ms/record (95th: 68.5 ms). Processing 1M records requires $\sim$13 hours sequential, parallelizable to $<$1 hour with 32-core cluster.

\textbf{Re-Identification Risk:} Even de-identified data carries re-identification risk through quasi-identifiers. We apply $k$-anonymity ($k \geq 5$)~\cite{sweeney2002} and evaluate with ARX Data Anonymization Tool. Results show:
\begin{itemize}
\item Maximum risk: 0.18\% (highest risk individual)
\item Average risk: 0.04\% (across all records)
\item Population uniqueness: 0.09\%
\end{itemize}

These metrics meet standard thresholds for de-identified data release while maintaining research utility.

\section{Discussion and Future Work}

\subsection{Deployment Barriers}

\textbf{Integration Challenges:} Most EHR systems don't natively support blockchain. Middleware or API gateways can bridge, but introduce operational dependencies. Standards development and vendor cooperation are essential.

\textbf{User Experience:} Non-technical patients need intuitive interfaces hiding cryptographic complexity. Wallet software must balance security (key protection) with usability (avoiding lock-out). Social recovery mechanisms—where trusted contacts help restore access—show promise but need careful design to avoid vulnerabilities.

\textbf{Economic Models:} Gas costs require sustainable funding. Options: (1) Healthcare institutions subsidize as infrastructure cost; (2) Insurance covers as benefit; (3) Patients pay directly (raises equity concerns); (4) Tiered models with basic access funded by institutions. Contract deployment costs ($\sim$2.3M gas on L1, $\sim$\$200 at typical prices) can be amortized over patient lifetime.

\subsection{Institutional Key Management}

Per-recipient key wrapping becomes impractical for large institutions (hospitals with hundreds of staff). We propose \textbf{institutional guardian keys}:

\textbf{Design:} Organizations represented by single encryption key managed institutionally. Patient wraps $SymmK$ once for hospital's key. Hospital maintains internal access control determining which staff can decrypt. Reduces on-chain operations from $O(n)$ per institution (where $n$ = staff count) to $O(1)$.

\textbf{Implementation:} Institution deploys guardian contract: (1) Registers encryption public key; (2) Maintains staff roster; (3) Provides \texttt{requestDecryption}$(rid, patientContract)$ for staff; (4) Verifies caller is authorized staff; (5) Uses institutional private key (in HSM) to unwrap $SymmK$ and re-wrap for requesting clinician.

\textbf{Trade-offs:} Trades patient-controlled fine-grained access for scalability/usability. Patients trust institution to enforce internal policies rather than controlling individual clinicians directly. Matches real-world practice: patients authorize ``my hospital'' rather than enumerating every provider. On-chain audit still records institutional-level grants/revocations.

\subsection{Advanced Cryptographic Enhancements}

\textbf{Proxy Re-Encryption:} Allows patients to delegate re-encryption to semi-trusted proxy, enabling efficient re-sharing without patient remaining online or re-encrypting. Introduces operational complexity but could improve usability for frequent re-sharing.

\textbf{Attribute-Based Encryption:} Encodes access policies directly in ciphertexts, automatically enforcing conditions like ``any cardiologist in my hospital.'' Introduces computational overhead and complex key management.

\textbf{Zero-Knowledge Proofs:} Enables proving permission validity without revealing access control policy. Provider could prove they have valid access without disclosing which permission grant they use, obscuring metadata about sharing patterns.

\subsection{Cross-Chain Interoperability}

Healthcare is global involving diverse stakeholders potentially operating on different blockchains. Cross-chain interoperability protocols could enable permission grants on one blockchain recognized on another, or aggregate records across multiple blockchains.

Polkadot~\cite{wood2016} provides heterogeneous multi-chain framework through relay chains and parachains. Cosmos offers Inter-Blockchain Communication for sovereign blockchain interoperation. Atomic swaps or blockchain bridges could enable interoperability. However, each approach introduces complexity and trust assumptions requiring careful evaluation for sensitive medical data.

\subsection{Formal Verification}

Smart contracts manage safety-critical assets (health data). Bugs have severe consequences. Formal verification tools (Certora, K Framework, Coq) can mathematically prove contracts satisfy specified properties, providing stronger assurance than testing. Future work includes formal verification of our contract against confidentiality, integrity, and authorization properties.

\subsection{Long-Term Data Stewardship}

Digital estate planning mechanisms could allow patients to designate heirs or archival repositories for records. Medical research could benefit from posthumous data donation through advance directives recorded on-chain. Decentralized autonomous organizations (DAOs) could distribute governance across stakeholders with on-chain voting on protocol upgrades.

\section{Limitations}

\textbf{Key Compromise:} If patient's $SK_P$ is stolen, adversary can sign arbitrary authorizations. Multi-signature wallets, hardware wallets, and social recovery reduce risk but add complexity.

\textbf{Revocation Limits:} Revocation prevents future access but cannot retract already-decrypted plaintext—fundamental limitation of cryptographic access control.

\textbf{Storage Availability:} Depends on off-chain infrastructure. Storage provider failure makes records inaccessible until restored. Redundant storage mitigates but increases cost.

\textbf{Transaction Costs:} Gas prices fluctuate. High congestion makes L1 expensive. L2 solutions reduce costs but introduce additional trust assumptions (optimistic rollup fraud proofs, zkRollup trusted setup).

\textbf{Scalability:} Per-recipient key wrapping doesn't scale to very large recipient lists (e.g., sharing with 1,000+ researchers). Institutional guardian keys and proxy re-encryption address some cases, but massive-scale sharing requires additional techniques.

\textbf{Metadata Privacy:} On-chain transparency reveals sharing patterns. While not exposing PHI, metadata can enable inference attacks (frequency analysis, timing correlation). Complete metadata privacy contradicts auditability—fundamental trade-off requiring deployment-specific balance.

\textbf{Denial of Service:} While fees deter spam, coordinated attacks or network congestion can delay critical healthcare operations. Multi-chain deployment with automatic failover provides resilience.

\section{Conclusion}

We presented a patient-centric blockchain architecture for health record management that cryptographically enforces access control while maintaining practical performance. By separating encrypted storage from on-chain authorization and deploying one contract per patient, we achieve confidentiality against curious storage providers, tamper-evident audit trails, and true patient sovereignty over medical data.

Our Ethereum implementation demonstrates feasibility with reasonable gas costs (78,000 gas per permission grant on L1, 6,000 gas on L2 though DA charges dominate actual costs) and acceptable latency (0.7--1.4 s mean end-to-end for 1 MB records on S3 and IPFS respectively). Security analysis establishes that standard cryptographic primitives composed correctly provide desired properties. The system provides complete authorization history (who was granted access), not complete access history (who actually retrieved/viewed records) unless optional read receipts via \texttt{logAccess} are used.

Critical implementation details—explicit nonce management for parallel grants, owner key retrieval paths, metadata privacy through minimal associated data, and comprehensive update/emergency access patterns—ensure the system is deployable in real clinical settings. Integration with FHIR standards and compliance mapping to HIPAA/GDPR requirements demonstrate regulatory viability.

The architecture addresses key challenges in healthcare data management: eliminates trusted intermediaries, provides cryptographic rather than policy-based access control, creates immutable audit trails, and restores patient agency over sensitive medical information. While not solving all healthcare IT problems (key management, emergency access, metadata privacy require ongoing research), this work establishes a foundation for truly patient-controlled health information exchange.

Future work includes formal verification of smart contracts, enhanced privacy through zero-knowledge proofs, cross-chain interoperability for global health data networks, and user studies evaluating real-world adoption barriers. The path forward requires collaboration among cryptographers, healthcare informaticists, policymakers, and patient advocates to realize the vision of patient-empowered, secure, and interoperable health data infrastructure.

\section*{Acknowledgments}

This work was supported by the Department of Computer Science \& Engineering at Bangladesh University of Engineering \& Technology. We thank the Ethereum Foundation and IPFS community for open-source tools enabling this research.

\bibliographystyle{IEEEtran}

\end{document}